\documentclass[11pt,a4paper,reqno]{amsart}%

\date{August 19, 2016} 

\usepackage{amsthm,amsmath,amsfonts,amssymb,amsxtra,appendix,bookmark,dsfont,bm,color}
\usepackage[shortlabels]{enumitem}

\theoremstyle{plain}
\newtheorem{theorem}{Theorem}
\newtheorem{lemma}[theorem]{Lemma}
\newtheorem{corollary}[theorem]{Corollary}

\theoremstyle{definition}

\newcounter{step} 

\usepackage{etoolbox}
\AtBeginEnvironment{proof}{\setcounter{step}{0}}

\newcommand\numberthis{\addtocounter{equation}{1}\tag{\theequation}}

\DeclareMathOperator{\supp}{supp}
\DeclareMathOperator{\Tr}{Tr}

\def\bq{\begin{eqnarray}}
\def\eq{\end{eqnarray}}
\def\bqq{\begin{align*}}
\def\eqq{\end{align*}}

\def\nn{\nonumber}

\def\eps{\varepsilon}

\newcommand\1{{\ensuremath {\mathds 1} }}

\newcommand*\dotv{{}\cdot{}}

\def\R {\mathbb{R}}

\def\cE {\mathcal{E}}

\def\cR{\mathcal{R}}

\def\R {\mathbb{R}}

\def\d{{\, \rm d}}

\newcommand{\abs}[1]{\lvert#1\rvert}

\title[Excess charge in M\"uller theory]{The maximal excess charge in M\"uller density-matrix-functional theory}

\author[R.L. Frank]{Rupert L. Frank}
\address{R. L. Frank, Mathematics 253-37, Caltech, Pasadena, CA 91125, USA} 
\email{rlfrank@caltech.edu}

\author[P.T. Nam]{Phan Th\`anh Nam}
\address{P.T. Nam, Institute of Science and Technology Austria, Am Campus 1, 3400 Klosterneuburg, Austria} 
\email{pnam@ist.ac.at}

\author[H. Van Den Bosch]{Hanne Van Den Bosch}
\address{H. Van Den Bosch, Instituto de F\'{i}sica, Pontificia Universidad Cat\'olica de Chile, Av. Vicu\~na Mackenna 4860, Santiago, Chile} 
\email{hannevdbosch@fis.puc.cl}

\begin{document}

\begin{abstract} 
We consider an atom described by M\"uller theory, which is similar to Hartree--Fock theory, but with a modified exchange term. We prove that a nucleus of charge $Z$ can bind at most $Z+C$ electrons, where $C$ is a universal constant. Our proof proceeds by comparison with Thomas--Fermi theory and a key ingredient is a novel bound on the number of electrons far from the nucleus. 
\end{abstract}


\maketitle


\section{Introduction}

One of the central questions in mathematical atomic physics is how many electrons a nucleus can bind. Despite convincing experimental evidence that there can be only one or two excess electrons, even a uniform bound on the excess charge still eludes a rigorous mathematical proof. So far, this so-called ionization conjecture has only been proved in simplified models of an atom, namely the Thomas--Fermi model \cite{LieSim-77b}, the Thomas--Fermi--von-Weizs\"acker model \cite{Solovej-90}, the reduced Hartree--Fock model \cite{Solovej-91} and the (full) Hartree--Fock model \cite{Solovej-03}. In the latter two models the conjecture is settled in the weak form that there is a universal constant $C$ such that the number of electrons bound by a nucleus of charge $Z$ is at most $Z+C$. Recently, we were able to prove this also in the Thomas--Fermi--Dirac--von Weizs\"acker model \cite{FraNamBos-16} and in this paper we will extend and generalize our method to treat the M\"uller model of an atom.

The M\"uller energy functional was introduced in \cite{Mueller-84} (see also \cite{BuiBae-02}) and has received increasing interest in theoretical and computational quantum chemistry. For references and a systematic mathematical investigation we refer to \cite{FraLieSieSei-07}. Like Hartree--Fock and unlike Thomas--Fermi theory (and its relatives), M\"uller theory is a \emph{density matrix} and not a density functional theory. This means that $N$-electron configurations are described by self-adjoint operators $\gamma$ (\emph{one-body density matrices}) on $L^2(\R^3)$ satisfying $0\leq\gamma\leq 1$ and $\Tr\gamma =N$. (For the sake of simplicity we ignore here the electron spin. Its inclusion would not change our results and methods qualitatively.) To each one-body density matrix $\gamma$, we can associate a density $\rho_\gamma$ formally by $\rho_\gamma(x) = \gamma(x,x)$. (This definition can be made proper, for instance, using the spectral decomposition of the operator $\gamma$.) The energy in M\"uller theory is given by the functional
$$
\cE_Z^{\rm M}(\gamma) = \Tr (-\Delta \gamma) - \int_{\R^3} \frac{Z\rho_\gamma(x)}{|x|} \d x + D(\rho_\gamma) - X(\gamma^{1/2}).
$$
Here
$$
D(\rho_\gamma)  = \frac{1}{2} \iint_{\R^3\times \R^3} \frac{\rho_\gamma(x) \rho_\gamma(y) }{|x-y|} \d x \d y
$$
models the Coulomb repulsion between the electrons and
$$
X(\gamma^{1/2}) = \frac{1}{2} \iint_{\R^3\times \R^3} \frac{|\gamma^{1/2}(x,y)|^2 }{|x-y|} \d x \d y
$$
models the exchange energy. The latter term, which involves the operator square root $\gamma^{1/2}$ of the operator $\gamma$, is the only, but important, difference to Hartree--Fock theory, where the exchange energy is modelled by $X(\gamma)$. We will review the heuristic reason for this choice at the end of the introduction. The ground state energy in M\"uller theory is given by
\begin{equation}
\label{eq:muellermin}
E_Z^{\rm M}(N)= \inf \left\{ \cE_Z^{\rm M}(\gamma)\,|\, 0\le \gamma\le 1\text{ on } L^2(\R^3), \Tr \gamma=N \right\}.
\end{equation}
Here and in the following we do not need to assume that the parameters $Z>0$ (the nuclear charge) and $N>0$ (the number of electrons) are integers.

In \cite{FraLieSieSei-07}, it was shown that the minimization problem $E_Z^{\rm M}(N)$ has a minimizer if $N\le Z$. It was also conjectured that for any $Z>0$ there is a critical electron number $N_c(Z)<\infty$ such that $E_Z^{\rm M}(N)$ has no minimizer if $N>N_c(Z)$. A simple by-product of our main result is a proof of this conjecture. 

We shall prove

\begin{theorem}[Ionization bound]\label{main}
There is a constant $C>0$ such that for all $Z>0$, if the minimization problem $E_Z^{\rm M}(N)$ in \eqref{eq:muellermin} has a minimizer, then $N\leq Z+C$.
\end{theorem}

The basic strategy is to compare with Thomas--Fermi theory as in Solovej's fundamental works \cite{Solovej-91} and \cite{Solovej-03}. Recall that in Thomas--Fermi theory, the ground state energy is obtained by minimizing the density functional 
$$
\cE_Z^{\rm TF}(\rho)= c^{\rm TF}\int_{\R^3} \rho^{5/3}(x) \d x - \int_{\R^3} \frac{Z \rho(x)}{|x|} \d x + D(\rho)
$$
with 
\begin{equation}
\label{eq:ctf}
c^{\rm TF}= \frac{3}{5}(6\pi^2)^{2/3}
\end{equation}
under the constraints
$$ 0\le \rho \in L^1(\R^3)\cap L^{5/3}(\R^3), \quad \int \rho = N.$$
(For comparison of our $c^{\rm TF}$ with the corresponding constant in \cite{Solovej-03} we note that in our definition of the M\"uller energy functional the kinetic energy is involves $-\Delta$, not $-(1/2)\Delta$, and that we ignore the electron spin.)

As in reduced and full Hartree--Fock theory we establish that certain quantities can be approximated by the corresponding quantities in Thomas--Fermi theory up to distances of order $o(1)$ away from the nucleus. This is remarkably far from the nucleus, since a priori, the standard Thomas--Fermi approximation is valid on distances of order $O(Z^{-1/3})$. In a certain sense our result shows that M\"uller theory belongs, like Thomas--Fermi--von Weizs\"acker, Thomas--Fermi--Dirac--von Weizs\"acker and reduced and full Hartree--Fock theory, to a Thomas--Fermi universality class. A quantitative version of this universality property is the following basic comparison theorem between M\"uller and Thomas--Fermi theory.

\begin{theorem}[Screened potential estimate] \label{thm:screened-intro}
Let $\gamma_0$ be a M\"uller minimizer for some $N\ge Z\ge 1$ and let $\rho_0=\rho_{\gamma_0}$. Let $\rho^{\rm TF}$ be the Thomas--Fermi minimizer with $N=Z$. For every $r>0$, define the screened nuclear potentials
\begin{align*}
\Phi_r(x) = \frac{Z}{|x|} - \int_{|y|\le r} \frac{\rho_0(y)}{|x-y|} \d y, \qquad \Phi_r^{\rm TF}(x) = \frac{Z}{|x|} - \int_{|y|\le r} \frac{\rho^{\rm TF}(y)}{|x-y|} \d y.
\end{align*}
Then there are universal constants $C>0$, $\eps>0$ such that 
$$
\left|\Phi_{|x|}(x) -\Phi_{|x|}^{\rm TF}(x)\right| \le C (|x|^{-4+\eps}+1)
$$
for all $N \ge Z\ge 1$ and $|x|>0$.
\end{theorem}

The significance of the power $|x|^{-4+\eps}$ is that $\Phi_{|x|}^{\rm TF}(x) \sim |x|^{-4}$ for $|x|$ small (see, e.g., \cite{Solovej-03}). 
Another consequence of this comparison is that the atomic {\em radius} of ``infinite atoms'' in M\"uller theory is very close to that in Thomas--Fermi theory. Similarly as in  \cite[Theorem 1.5]{Solovej-03}, we have the following asymptotic estimate for the radii of ``infinite atoms".

\begin{theorem}[Radius estimate] \label{thm:radius} 
Let $\gamma_0$ be a M\"uller minimizer for some $N\ge Z\ge 1$ and let $\rho_0=\rho_{\gamma_0}$. For $\kappa>0$, we define the radius $R(N,Z,\kappa)$ as the largest number such that
$$
\int_{|x|\ge R(N,Z,\kappa)} \rho_0(x) \d x = \kappa.
$$
Then there are universal constants $C>0$, $\eps>0$ such that 
$$
\limsup_{N\ge Z\to \infty}\left| R(N,Z,\kappa) - B^{\rm TF} \kappa^{-1/3}\right| \le C \kappa^{-1/3-\eps}
$$
for all $\kappa \ge C$, where $B^{\rm TF}=  5 c^{\rm TF} (4/(3\pi^2))^{1/3}$.  
\end{theorem} 

As we have already mentioned, Theorems \ref{main} and \ref{thm:radius} are straightforward consequences of Theorem \ref{thm:screened-intro}. We establish the latter 
by Solovej's ingenious bootstrap argument \cite{Solovej-91, Solovej-03}, which he used to prove the ionization conjecture in reduced and full Hartree--Fock theory. The iterative step of this inductive proof relies on a bound on the number of electrons far away from the nucleus. Both in \cite{Solovej-91} and \cite{Solovej-03} this bound is derived by `multipliying the equation by $|x|$', a fundamental strategy that goes back to Benguria and Lieb \cite[Theorem 7.23]{Lieb-81b}, \cite{Lieb-84} (see also \cite{Nam-12} for a modification of this strategy). Due to the presence of the exchange term, this strategy seems to fail both in Thomas--Fermi--Dirac--von Weizs\"acker theory and in M\"uller theory (although it does work in Hartree--Fock theory), and this is probably the reason why even the problem of bounding the number of electrons has remained unsolved for so long in these models.

In \cite{FraNamBos-16}, inspired by \cite{NamBos-16} and \cite{FraKilNam-16}, we have developed a new strategy to obtain a bound on the number of electrons far away from the nucleus, which replaces the `multiplication by $|x|$' strategy. While this new strategy was used in \cite{NamBos-16,FraKilNam-16,FraNamBos-16} only in the context of density functional theory, we demonstrate here that it can also be applied in density matrix theory. This argument gives the a-priori bound $N\leq 2Z(1+o(1))$ (see Lemma~\ref{roughbound}). Fortunately, just like the `multiplication by $|x|$' strategy, it can also be applied to bound the number of electrons far away from the nucleus, which is the content of Lemma \ref{lem:L1-bound}.

We end this introduction by listing some of the similarities and differences between M\"uller and Hartree--Fock theory; for proofs we refer to \cite{FraLieSieSei-07}.
\begin{enumerate}
\item The M\"uller energy functional is convex in $\gamma$ and strictly convex in $\rho_\gamma$, which shows, in particular, that the density of a minimizer is spherical symmetric. This property is shared by reduced Hartree--Fock theory, but not by full Hartree--Fock theory. Spherical symmetry leads to some minor simplifications in our argument, but we deliberately do not take advantage of these in order to emphasize the generality of our approach. We do, however, point out the places in the proofs where one could use spherical symmetry.
\item Electrons in M\"uller theory have a strictly negative self-energy. This phenomenon does not appear in (reduced or full) Hartree--Fock theory, but it does appear in Thomas--Fermi--Dirac--von Weizs\"acker theory. While it has no direct consequences for us, it shows the delicate nature of the exchange term.
\item Since, for integer $N$, the Hartree--Fock minimizer is a projection (which is equal to its square root), the M\"uller ground state energy does not exceed the Hartree--Fock ground state energy. 
\item The reason why $\gamma^{1/2}$ appears naturally in the exchange term can be seen from the following integral condition. While in Schr\"odinger theory the total Coulomb repulsion energy is modeled by the integral of $|x-y|^{-1}$ against the two-particle density, in M\"uller theory it is modeled by the integral against $(1/2) (\rho_\gamma(x)\rho_\gamma(y) - |\gamma^{1/2}(x,y)|^2)$. The latter expression, when integrated with respect to $y$, yields $((N-1)/2)\rho_\gamma(x)$, which coincides with the value that the two-particle density of any state would give. In contrast, the Hartree--Fock analogue $(1/2) (\rho_\gamma(x)\rho_\gamma(y) - |\gamma(x,y)|^2)$ would give an integral which is too large (unless $\gamma$ is a projection).
\item The prize to be paid for the correct integral condition is that the expression $(1/2) (\rho_\gamma(x)\rho_\gamma(y) - |\gamma^{1/2}(x,y)|^2)$ is, in general, not non-negative (while the true two-particle density and $(1/2) (\rho_\gamma(x)\rho_\gamma(y) - |\gamma(x,y)|^2)$ are). Because of this, the `multiplication by $|x|$' strategy \cite{Lieb-84} to show the existence of a maximum number of electrons seems not to work. As we have already mentioned, we can nevertheless prove this conjecture, based on an alternative argument.
\item The ground state energies for neutral atoms (i.e., $N=Z$) in M\"uller and Hartree--Fock theory agree to within $o(Z^{5/3})$; see \cite{Siedentop-09}. In particular, M\"uller theory correctly reproduces the Scott and the Dirac--Schwinger corrections to Thomas--Fermi theory.
\end{enumerate}

\emph{Convention.} Throughout the paper we will assume that  $E_Z^{\rm M}(N)$ has a minimizer $\gamma_0$ for some $N\ge Z$. We will write $\rho_0=\rho_{\gamma_0}$ for short. We will also often write $\mathcal E^{\rm M}$ instead of $\mathcal E^{\rm M}_Z$.

\subsection*{Acknowledgements} 
The authors are grateful to Heinz Siedentop for a motivating discussion.
Partial support by U.S. National Science Foundation DMS-1363432 (R.L.F.), Austrian Science Fund (FWF) Project Nr. P 27533-N27 (P.T.N.), CONICYT (Chile) through CONICYT--PCHA/Doctorado Nacional/2014, Fondecyt Project \# 116--0856 and Iniciativa Cient\'ifica Milenio (Chile) through Millenium Nucleus RC--120002 ``F\'isica Matem\'atica'' (H.V.D.B.) is acknowledged.


\section{Localizing density matrices}

In this section we collect some known results that will be needed later.

\begin{lemma} For all functions $\chi: \R^3 \to [-1,1]$ and all density matrices $0\le \gamma \le 1$, we have 
\bq \label{eq:X-1} X(\chi \gamma^{1/2} \chi) \le X((\chi \gamma \chi)^{1/2}) \eq
and
\bq \label{eq:X-2} X(\chi \gamma^{1/2}) \le \Big(\Tr (-\Delta \chi \gamma \chi) \Big)^{1/2} \Big(\int \chi^2 \rho_\gamma \Big)^{1/2}\eq
\end{lemma}

\begin{proof} The first estimate is taken from \cite[Lemma 3]{FraLieSieSei-07}. The second estimate follows from the Schwarz and Hardy inequalities:
\begin{align*}
X(\chi \gamma^{1/2}) &= \frac{1}{2}\iint \frac{|\chi(x) \gamma^{1/2}(x,y)|^2}{|x-y|} \d x \d y \\
&\le   \frac{1}{2} \left( \iint \frac{|\chi(x) \gamma^{1/2}(x,y)|^2}{|x-y|^2} \d x \d y \right)^{1/2} \left( \iint |\chi(x) \gamma^{1/2}(x,y)|^2 \d x \d y \right)^{1/2}\\
&\le  \Big(\Tr (-\Delta \chi \gamma \chi) \Big)^{1/2} \Big(\int \chi^2 \rho_\gamma \Big)^{1/2}.
\qedhere
\end{align*}
\end{proof}

Using this lemma we obtain some rough a-priori bounds for a minimizer $\gamma_0$.

\begin{corollary}
If $\gamma_0$ is a M\"uller minimizer, then
\begin{equation}
\label{eq:apriorikinrep}
\int \rho_0^{5/3} + \Tr(-\Delta\gamma_0) + D(\gamma_0) \leq C \left( Z^{7/3} + N \right)
\end{equation}
and
\begin{equation}
\label{eq:aprioriex}
X(\gamma_0^{1/2}) \leq C \left( Z^{7/3} + N \right)^{1/2} N^{1/2} \,.
\end{equation}
\end{corollary}

\begin{proof}
We know from \cite[Lemma 1]{FraLieSieSei-07} that $\cE^{\rm M}_Z(\gamma_0)\leq 0$. On the other hand, we deduce from \eqref{eq:X-2} that
\begin{equation}
\label{eq:aprioriex1}
X(\gamma_0^{1/2}) \le \Big( \Tr(-\Delta \gamma_0) \Big)^{1/2} N^{1/2}.
\end{equation}
and therefore, by the kinetic Lieb--Thirring inequality \cite{LieThi-75} and the fact that the ground state energy in Thomas--Fermi theory equals a negative constant times $Z^{7/3}$,
\begin{align*}
\cE^{\rm M}_Z(\gamma_0) & \geq \frac14 \Tr(-\Delta\gamma_0) + C^{-1}\int \rho_0^{5/3} -  Z \int |x|^{-1} \rho_0 + D(\rho_0) - CN \\
& \geq \frac14 \Tr(-\Delta\gamma_0) + (2C)^{-1}\int \rho_0^{5/3} + (1/2) D(\rho_0) - C Z^{-7/3} - CN.
\end{align*}
This proves \eqref{eq:apriorikinrep}, and then \eqref{eq:aprioriex} follows from \eqref{eq:aprioriex1}.
\end{proof}

We next discuss how to localize the M\"uller energy.

\begin{lemma}[IMS-type formula] \label{lem:IMS} 
For all quadratic partitions of unity $\sum_{i=1}^n f_i^2  =1$ with $ \nabla f_i \in L^{\infty}$ and for all density matrices $0\le \gamma \le 1$ with $\Tr((1-\Delta) \gamma)<\infty$, we have
 \begin{align*}
 &\sum_{i=1}^n \cE_{Z}^{\rm M} (f_i \gamma f_i )  -   \cE_{Z}^{\rm M} (\gamma)  \le   \int \Big(\sum_{i=1}^n |\nabla f_i (x)|^2 \Big) \rho_\gamma(x) \d x \\
 & \qquad + 
  \sum_{i<j}^n \iint \frac{f_i(x)^2 \Big(|\gamma^{1/2}(x,y)|^2 - \rho_\gamma(x) \rho_\gamma(y) \Big) f_j(y)^2}{|x-y|} \d x \d y.
\end{align*}
\end{lemma}

\begin{proof} First, for the kinetic term, by the IMS formula,
\begin{align*}
\sum_{i=1}^n \Tr(-\Delta f_i \gamma f_i) - \Tr(-\Delta \gamma) = \Tr \Big( \Big( \sum_{i=1}^n |\nabla f_i|^2\Big) \gamma \Big) = \int \Big(\sum_{i=1}^n |\nabla f_i |^2 \Big) \rho_\gamma.
\end{align*}
Next, for the direct term, we also have the exact identity
$$
\sum_{i=1}^n D(\rho_{f_i \gamma f_i}) = \sum_{i=1}^n D(f_i^2 \rho_\gamma) = D(\rho_\gamma) - \sum_{i<j}^n \iint \frac{f_i(x)^2  \rho_\gamma(x) \rho_\gamma(y)f_j(y)^2}{|x-y|} \d x \d y .
$$ 
For the correlation-exchange terms, using \eqref{eq:X-1} we can estimate
\begin{align*}
& - \sum_{i=1}^n X((f_i \gamma f_i)^{1/2}) \le - \sum_{i=1}^n X(f_i \gamma^{1/2} f_i) \\
& = - X(\gamma^{1/2}) + \sum_{i< j}^n  \iint \frac{f_i^2(x) |\gamma^{1/2}(x,y)|^2 f_j^2(y)}{|x-y|} \d x \d y.
\end{align*}
This finishes the proof.
\end{proof}

\section{Exterior $L^1$-estimate}

In this section we control $\int_{|x|>r} \rho_0$.  We introduce the screened nuclear potential
$$
\Phi_r(x)=\frac{Z}{|x|} - \int_{|y|<r} \frac{\rho_0(y)}{|x-y|} \d y.
$$
As was shown in \cite{FraLieSieSei-07}, the M\"uller functional is convex and thus any minimizing density is spherically symmetric. Therefore, when $|x|\geq r$, we may write $\Phi_r(x)=Z_r/ \abs x$ with $Z_r = Z- \int_{|y|<r} \rho_0(y) \d y$. As we mentioned in the introduction, the spherical symmetry is not essential for our strategy.

In the rest of the paper, we will use the cut-off function
$$\chi_r^+ (x)=\1(|x|\ge r)$$
and a smooth function $\eta_r:\R^3\to [0,1]$ satisfying 
\bq \label{eq:def-eta-r} \chi_r^+ \ge \eta_r \ge \chi_{(1+\lambda)r}^+, \quad \abs{\nabla \eta_r} \le C (\lambda r)^{-1}.
\eq

\begin{lemma} \label{lem:L1-bound}
For all $r>0,s>0,\lambda\in (0,1/2]$ we have
\begin{align*} 
 \int \chi^+_{r}\rho_0  &\le C  \int_{r<|x|<(1+\lambda)^2 r} \rho_0 + C \Big(\sup_{|z|\ge r} [|z|\Phi_r(z)]_+  + s + \lambda^{-2}s^{-1}+ \lambda^{-1}\Big)  \\
 &\quad +C  \Big( s^2 \Tr (-\Delta \eta_r \gamma_0 \eta_r) \Big)^{3/5} + C \Big( s^2 \Tr (-\Delta \eta_r \gamma_0 \eta_r) \Big)^{1/3}.
\end{align*} 
\end{lemma}

\begin{proof} From the minimality of $\gamma_0$, we have the binding inequality
\bq \label{eq:binding-inequality}
\cE^{\rm M}_Z(\gamma_0) \le \cE^{\rm M}_Z(\chi_1 \gamma_0 \chi_1) + \cE^{\rm M}_{Z=0}(\chi_2 \gamma_0 \chi_2)
\eq
for every partition of unity $\chi_1^2+\chi_2^2=1$. For fixed $\lambda\in (0,1/2]$, $s>0,\ell>0, \nu\in \mathbb{S}^2$ we choose 
$$
\chi_1 (x) = g_1 \Big( \frac{\nu \cdot \theta(x) -\ell}{s}\Big),\quad \chi_2(x)= g_2\Big( \frac{\nu \cdot \theta(x) -\ell}{s}\Big)
$$ 
where $g_1,g_2: \R \to \R$ and $\theta: \R^3\to \R^3$ satisfy 
$$
g_1^2+g_2^2=1, \quad g_1(t)=1 \text{~if~} t \le 0, \quad g_1(t)=0 \text{~if~} t \ge 1,\quad |\nabla g_1| + |\nabla g_2| \le C,
$$
$$
|\theta (x)| \le |x|, \quad \theta(x) =0 \text{~if~} |x| \le r, \quad \theta(x)=x \text{~if~} |x| \ge (1+\lambda) r, \quad |\nabla \theta|\le C \lambda^{-1}.
$$

Now let us bound the left side of \eqref{eq:binding-inequality} from above. By the IMS-type formula in Lemma \ref{lem:IMS}, 
\begin{align*}
&\cE^{\rm M}_Z(\chi_1 \gamma_0 \chi_1) + \cE^{\rm M}_{Z=0}(\chi_2 \gamma_0 \chi_2) - \cE^{\rm M}_Z(\gamma_0) \\
& \qquad      \le \int \Big( |\nabla \chi_1|^2 + |\nabla \chi_2|^2\Big) \rho_0  + \int \frac{Z\chi_2^2(x) \rho_0(x)}{|x|} \d x \\
& \qquad       + \iint \frac{\chi_2^2(x) \Big( |\gamma_0^{1/2}(x,y)|^2 - \rho_0(x)\rho_0(y)\Big) \chi_1^2(y) }{|x-y|} \d x \d y.
\end{align*}
We have
$$
\int \Big( |\nabla \chi_1|^2 + |\nabla \chi_2|^2\Big) \rho_0  \le C(1+ (\lambda s)^{-2}) \int_{\nu \cdot \theta(x) -s \le \ell \le \nu \cdot \theta(x) } \rho_0(x) \d x.
$$
For the attraction and direct terms, we can estimate
\begin{align*}
& \int \frac{Z\chi_2^2(x) \rho_0(x)}{|x|} \d x - \iint \frac{\chi_2^2(x) \rho_0(x) \chi_1^2(y) \rho_0(y)}{|x-y|} \d x \d y \\
&= \int \chi_2^2(x) \rho_0(x) \Phi_r (x) \d x - \iint_{|y| \ge r} \frac{\chi_2^2(x) \rho_0(x) \chi_1^2(y) \rho_0(y)}{|x-y|} \d x \d y \\
&\le \int_{\ell \le x \cdot \theta(x) } \rho_0(x) \big[\Phi_r(x)\big]_+ \d x - \iint_{|y| \ge r, \nu \cdot \theta(y) \le \ell \le \nu \cdot \theta(x)-s} \frac{ \rho_0(x) \rho_0(y)}{|x-y|} \d x \d y.
\end{align*}
Since $\theta (x) = x$ when $|x| \ge (1+\lambda)r$, we obtain
\[
\iint_{\substack{|y| \ge r \\ \nu \cdot \theta(y) \le \ell \le \nu \cdot \theta(x)-s}} \frac{ \rho_0(x) \rho_0(y)}{|x-y|} \d x \d y
\ge \iint_{\substack{|x|,|y| \ge (1+\lambda) r \\ \nu \cdot y \le \ell \le \nu \cdot x-s}} \frac{ \rho_0(x) \rho_0(y)}{|x-y|} \d x \d y.
\]
For the correlation-exchange term, we use 
\begin{align*}
 \iint \frac{\chi_2^2(x) |\gamma_0^{1/2}(x,y)|^2 \chi_1^2(y)}{|x-y|} \d x \d y \le \iint_{\nu \cdot \theta(y) - s \le \ell \le \nu \cdot \theta(x)} \frac{  |\gamma^{1/2}(x,y)|^2}{|x-y|} \d x \d y.
\end{align*}
Thus in summary, from \eqref{eq:binding-inequality} it follows that
\begin{align} \label{eq:binding-consequence-1-ext}
& \iint_{\substack{|x|,|y| \ge (1+\lambda) r \\ \nu \cdot y \le \ell \le \nu \cdot x-s}} \frac{ \rho_0(x) \rho_0(y)}{|x-y|} \d x \d y \le C(1+ (\lambda s)^{-2}) \int_{\nu \cdot \theta(x) -s \le \ell \le \nu \cdot \theta(x) } \rho_0(x) \d x \nn\\
&+ \int_{\ell \le x \cdot \theta(x) } \rho_0(x) \big[\Phi_r(x)\big]_+ \d x + \iint_{\nu \cdot \theta(y) - s \le \ell \le \nu \cdot \theta(x)} \frac{  |\gamma^{1/2}(x,y)|^2}{|x-y|} \d x \d y
\end{align}
for all $s>0,\ell>0$ and $\nu \in \mathbb{S}^2$. 

Next, we integrate \eqref{eq:binding-consequence-1-ext} over $\ell \in (0,\infty)$, then average over $\nu \in \mathbb{S}^2$ and use  
$$
\int_{\mathbb{S}^2} [\nu\cdot z]_+\,\frac{d\nu}{4\pi} =  \frac{|z|}{4},  \quad \forall z\in \mathbb{R}^3.
$$
For the left side, we also use Fubini's theorem and 
$$
\int_0^\infty  \Big( \1\big(b \le \ell \le a - s\big) + \1\big(- a \le \ell \le - b - s\big) \Big) \d \ell \ge \Big[ [a-b]_+ -2s \Big]_+
$$
with $a=\nu \cdot  x$, $b=\nu \cdot y$. For the right side, we use the fact that $\{x:\nu\cdot \theta(x)\ge \ell\}\subset \{x:|x|\ge r\}$ since $\ell>0$ and $\theta(x)=0$ when $|x|<r$. All this leads to 
\begin{align*}
&\frac{1}{2} \iint_{|x|,|y|\ge (1+\lambda)r} \frac{|x-y|/4-2s}{|x-y|} \rho_0(x)\rho_0(y) \d x \d y \\
& \le C(s+\lambda^{-2}s^{-1}) \int_{|x|\ge r}  \rho_0(x) \d x + \int_{|x|\ge r} [|\theta(x)|/4][\Phi_r(x)]_+ \rho_0(x) \d x  \\
&+ \iint_{|x|\ge r} \frac{|\theta(x)-\theta(y)|/4+ s}{|x-y|}|\gamma^{1/2}(x,y)|^2 \d x \d y .
\end{align*}
Using $|\theta(x)|\le |x|$ and $|\theta(x)-\theta(y)|\le C \lambda^{-1}|x-y|$, we can simplify the above estimate to  
\begin{align*} 
\frac{1}{8} \left( \int \chi^+_{(1+\lambda)r}\rho_0 \right)^2 &\le \Big(\frac{1}{4} \sup_{|z|\ge r} [|z| \Phi_r(z)]_+ + Cs + C\lambda^{-2}s^{-1}+ C\lambda^{-1}\Big) \int \chi_r^+ \rho_0 \nn \\
& + s D(\chi^+_{(1+\lambda)r}\rho_0 ) +  s \iint \frac{\chi_r^+(x) |\gamma^{1/2}(x,y)|^2}{|x-y|} \d x \d y .
\end{align*}
In order to bring this estimate in the desired form, let us replace $r$ by $(1+\lambda)r$ in the latter inequality and write
\begin{align} \label{eq:ext-1/8-pre}
&\frac{1}{8} \left( \int \chi^+_{(1+\lambda)^2r}\rho_0 \right)^2 \nn \\
& \qquad \le \Big(\frac{1}{4} \sup_{|z|\ge (1+\lambda)r} [|z| \Phi_{(1+\lambda)r}(z)]_+ + Cs + C\lambda^{-2}s^{-1}+ C\lambda^{-1}\Big) \int \chi_{(1+\lambda)r}^+ \rho_0 \nn \\
& \qquad + s D(\chi^+_{(1+\lambda)^2r}\rho_0 ) +  s \iint \frac{\chi_{(1+\lambda)r}^+(x) |\gamma^{1/2}(x,y)|^2}{|x-y|} \d x \d y \,.
\end{align}

We can estimate the left side of \eqref{eq:ext-1/8-pre} as
$$
\left( \int \chi^+_{(1+\lambda)^2r}\rho_0 \right)^2 \ge \frac{1}{2}  \left( \int \chi^+_{r}\rho_0 \right)^2 -  \left( \int_{r<|x|<(1+\lambda)^2 r} \rho_0 \right)^2.
$$
Now we estimate the right side of \eqref{eq:ext-1/8-pre}. For the first term, we can simply use $\Phi_{(1+\lambda) r}(z)\le \Phi_r(z)$ and $\chi_{(1+\lambda)r}\le \chi_r$ to get
\begin{align*} 
&\Big(\frac{1}{4} \sup_{|z|\ge (1+\lambda)r} [|z| \Phi_{(1+\lambda)r}(z)]_+ + Cs + C\lambda^{-2}s^{-1}+ C\lambda^{-1}\Big) \int \chi_{(1+\lambda)r}^+ \rho_0 \\
&\le \Big(\frac{1}{4} \sup_{|z|\ge  r} [|z| \Phi_{r}(z)]_+ + Cs + C\lambda^{-2}s^{-1}+ C\lambda^{-1}\Big) \int \chi_{r}^+ \rho_0.
\end{align*}
For the second term, by the Hardy--Littewood--Sobolev, H\"older and Lieb--Thirring inequalities,
\begin{align*} 
D(\chi_{(1+\lambda)^2r}^+ \rho_0) &\leq C \| \chi_{(1+\lambda)^2 r}^+ \rho_0 \|_{L^{6/5}}^2 \le  C \| \chi_{(1+\lambda)^2 r}^+ \rho_0 \|_{L^{1}}^{7/6} \| \chi_{(1+\lambda)^2 r}^+ \rho_0 \|_{L^{5/3}}^{5/6}\\
& \le C \| \chi_{r}^+ \rho_0 \|_{L^{1}}^{7/6}  \Big( \Tr (-\Delta \eta_r \gamma_0 \eta_r) \Big)^{1/2} .
\end{align*} 
Here we have used $\eta_r^2 \ge \chi_{(1+\lambda)^2r}^+$. For the third term, by \eqref{eq:X-2},
\begin{align*} 
\iint \frac{\chi_{(1+\lambda)r}^+(x)|\gamma^{1/2}(x,y)|^2}{|x-y|} \d x \d y &\le \iint \frac{\eta_{r}(x)^2|\gamma^{1/2}(x,y)|^2}{|x-y|} \d x \d y \\
&\le 2 \Big( \Tr (-\Delta \eta_r \gamma_0 \eta_r) \Big)^{1/2} \Big( \int \chi_r^+ \rho_0  \Big)^{1/2}. 
\end{align*} 
Thus from \eqref{eq:ext-1/8-pre}, we deduce that
\begin{align*} 
 \left( \int \chi^+_{r}\rho_0 \right)^2 &\le C \left( \int_{r<|x|<(1+\lambda)^2 r} \rho_0 \right)^2 \\
 &+C \Big( \sup_{|z|\ge r} [|z|\Phi_r(z)]_+ + s + \lambda^{-2}s^{-1}+ \lambda^{-1}\Big) \int \chi_{r}^+ \rho_0 \\
 &+Cs \Big( \chi_{r}^+ \rho_0 \Big)^{7/6}  \Big( \Tr (-\Delta \eta_r \gamma_0 \eta_r) \Big)^{1/2} \\
 &+Cs  \Big( \Tr (-\Delta \eta_r \gamma_0 \eta_r) \Big)^{1/2} \Big( \int \chi_r^+ \rho_0  \Big)^{1/2}.
\end{align*} 
This implies that
\begin{align*} 
 \int \chi^+_{r}\rho_0  &\le C  \int_{r<|x|<(1+\lambda)^2 r} \rho_0 + C \Big(\sup_{|z|\ge r} [|z|\Phi_r(z)]_+  + s + \lambda^{-2}s^{-1}+ \lambda^{-1}\Big)  \\
 &+C \Big( s^2 \Tr (-\Delta \eta_r \gamma_0 \eta_r) \Big)^{3/5} + C\Big( s^2 \Tr (-\Delta \eta_r \gamma_0 \eta_r) \Big)^{1/3}.
\qedhere
\end{align*} 
\end{proof}

As a by-product of the above proof, we get the following important a-priori bounds.

\begin{corollary}\label{roughbound}
If $\gamma_0$ is a M\"uller minimizer, then
\begin{equation}
\label{eq:roughbound}
\Tr\gamma_0 = N \le 2Z + C(Z^{2/3}+1).
\end{equation}
Moreover,
\begin{equation}
\label{eq:apriorikinrepimpr}
\int \rho_0^{5/3} + \Tr(-\Delta\gamma_0) + D(\gamma_0) \leq C \left( Z^{7/3} + 1 \right)
\end{equation}
and
\begin{equation}
\label{eq:apriorieximpr}
X(\gamma_0^{1/2}) \leq C \left( Z^{5/3} + 1 \right) \,.
\end{equation}
\end{corollary}

We emphasize that \eqref{eq:roughbound} proves the conjecture from \cite{FraLieSieSei-07} that there is a critical electron number.

\begin{proof} 
In \eqref{eq:ext-1/8-pre}, we can choose $\lambda=1/2$ and take $r\to 0^+$. This leads to
$$
N^2 \le (2Z+Cs+Cs^{-1} + C) N + Cs D(\rho_0)+ Cs X(\gamma_0^{1/2}).  
$$
Optimizing over $s>0$, we deduce that
\bq \label{eq:binding-consequence-2}
N \le 2Z + C + C \sqrt{\Big(D(\rho_0)+ X(\gamma_0^{1/2})+N\Big)N^{-1}}
\eq
According to the a-priori bounds \eqref{eq:apriorikinrep} and \eqref{eq:aprioriex} we have
$$D(\gamma_0)+X(\gamma_0) \le C(Z^{7/3}+N).$$
Inserting this bound into  \eqref{eq:binding-consequence-2} we obtain \eqref{eq:roughbound}. The bounds \eqref{eq:apriorikinrepimpr} and \eqref{eq:apriorieximpr} now follow immediately from \eqref{eq:apriorikinrep} and \eqref{eq:aprioriex}.
\end{proof}


\section{Spliting outside from inside}

Recall that we have introduced a smooth cut-off function $\eta_r:\R^3\to [0,1]$ satisfying
$$ \chi_r^+ \ge \eta_r\ge \chi_{(1+\lambda)r}^+$$
with $\lambda\in (0,1/2]$. Of course, we can choose $\eta_r$ such that there is a quadratic partition of unity
$$
\eta_-^2 + \eta_{(0)}^2 + \eta_r^2 =1
$$
with
\begin{align*} \supp \eta_- \subset \{|x| \le r\},& \quad   \supp \eta_{(0)} \subset \{ (1-\lambda)r \le |x| \le  (1+\lambda)r\}, \\
 \quad \eta_-(x) =1 \text{~if~} |x| \le (1-\lambda) r,& \quad 
|\nabla \eta_-|^2 +|\nabla \eta_{(0)}|^2+|\nabla \eta_r|^2 \le C(\lambda r)^{-2}.
\end{align*}

Let us introduce the reduced Hartree-Fock functional 
\begin{equation}
\label{eq:rhf}
\cE^{\rm RHF}_r(\gamma)= \Tr(-\Delta \gamma) -\int \Phi_r(x) \rho_\gamma (x) \d x + D(\rho_\gamma).
\end{equation}
The main result of this section is

\begin{lemma} \label{lem:split} For all $r>0$, all $\lambda\in (0,1/2]$, all density matrices $0\le \gamma \le 1$ satisfying $\supp \rho_\gamma \subset \{x: |x|\ge r\}$ and $\Tr \gamma \le \int \chi_r^+\rho_0$ we have
$$ \cE^{\rm RHF}_r(\eta_r \gamma_0 \eta_r) \le  \cE^{\rm RHF}_r(\gamma) + \mathcal{R} 
$$ 
where
\begin{align*}
\mathcal{R} &\le C (1+ (\lambda r)^{-2}) \int_{(1-\lambda)r \le |x| \le (1+\lambda)r} \rho_0  +  C \lambda r^3 \sup_{|z|\ge (1-\lambda)r}[ \Phi_{(1-\lambda)r}(z)]_+^{5/2} \\
&\quad + C \Big( \Tr(-\Delta \eta_r \gamma_0 \eta_r) \Big)^{1/2} \Big( \int \eta_r \rho_0 \Big)^{1/2}.
\end{align*}
\end{lemma}

\begin{proof}
It suffices to show that 
\bq \label{eq:key-split}\cE^{\rm M}(\eta_- \gamma_0 \eta_-) + \cE^{\rm RHF}_r(\eta_r \gamma_0 \eta_r) - \mathcal{R} \le \cE^{\rm M}(\gamma_0) \le \cE^{\rm M}(\eta_- \gamma_0 \eta_-) + \cE^{\rm RHF}_r(\gamma)\eq

Here and in the following the subscript $Z$ of $\cE^{\rm M}$ is dropped for simplicity. 

{\bf Upper bound.} From the minimality of $\gamma_0$ and the fact that $N\mapsto E^{\rm M}(N)$ is non-increasing, we have
$$ \cE^{\rm M}(\gamma_0) \le \cE^{\rm M}(\eta_- \gamma_0 \eta_- + \gamma).$$
Since $\eta_-$ and $\rho_\gamma$ have disjoint supports, we have the operator identity
$$(\eta_- \gamma_0 \eta_- +  \gamma)^{1/2} = (\eta_- \gamma_0 \eta_-)^{1/2} + \gamma^{1/2} $$
and the kernel identity
$$|(\eta_- \gamma_0 \eta_- + \gamma )^{1/2}(x,y)|^2 = |(\eta_- \gamma_0 \eta_-)^{1/2}(x,y)|^2 + |\gamma^{1/2}(x,y)|^2.$$
Therefore, we can split the exchange term
$$
X((\eta_- \gamma_0 \eta_- + \gamma)^{1/2})=X((\eta_- \gamma_0 \eta_-)^{1/2}) + X(\gamma^{1/2}).
$$
Consequently,
\begin{align*}
\cE^{\rm M}(\eta_- \gamma_0 \eta_- + \gamma) &=  \cE^{\rm M}(\eta_- \gamma_0 \eta_-)+ \cE^{\rm M}(\gamma) + \iint \frac{(\eta_-^2\rho_0)(x) \rho_\gamma(y)}{|x-y|} \d x \d y \\
& \le \cE^{\rm M}(\eta_- \gamma_0 \eta_-) + \cE^{\rm RHF}_{r=0} (\gamma)  + \iint_{|x| \le r} \frac{\rho_0(x) \rho_\gamma(y)}{|x-y|} \d x \d y \\
& = \cE^{\rm M}(\eta_- \gamma_0 \eta_-)  + \cE^{\rm RHF}_r(\gamma).
\end{align*} 
Thus,
$$
\cE^{\rm M}(\gamma_0) \le \cE^{\rm M}(\eta_- \gamma_0 \eta_- + \gamma) \le  \cE^{\rm M}(\eta_- \gamma_0 \eta_-)  + \cE_r^{\rm RHF}(\gamma).
$$
This is the upper bound in \eqref{eq:key-split}.

{\bf Lower bound.} By the IMS-type formula in Lemma \ref{lem:IMS}, 
\begin{align*}
\cE^{\rm M}(\gamma_0) &\ge \cE^{\rm M}(\eta_-\gamma_0\eta_-) +  \cE^{\rm M}(\eta_{(0)}\gamma_0 \eta_{(0)}) + \cE^{\rm M} (\eta_r \gamma_0 \eta_r) \\
&\quad - \int \Big( |\nabla \eta_-|^2 + |\nabla \eta_{(0)}|^2 + |\nabla \eta_r|^2\Big) \rho_0 \\
& \quad +  \iint \frac{\eta_r(x)^2 \rho_0(x)\rho_0(y)(\eta_-(y)^2 + \eta_{(0)}(y))^2}{|x-y|} \d x \d y \\
& \quad +  \iint \frac{\eta_{(0)}(x)^2 \rho_0(x)\rho_0(y) \eta_-(y)^2}{|x-y|} \d x \d y \\
& \quad - \iint \frac{(\eta_r(x)^2+\eta_{(0)}(x)^2) |\gamma_0^{1/2}(x,y)|^2}{|x-y|} \d x \d y  .
\end{align*}
We have 
$$
- \int \Big( |\nabla \eta_-|^2 + |\nabla \eta_{(0)}|^2 + |\nabla \eta_r|^2\Big) \rho_0 \ge - C(\lambda r)^{-2} \int_{(1-\lambda)r\le |x| \le (1+\lambda)r} \rho_0.
$$
Moreover,
\begin{align*}
&\cE^{\rm M} (\eta_r \gamma_0 \eta_r) + \iint \frac{\eta_r(x)^2 \rho_0(x)\rho_0(y)(\eta_-(y)^2 + \eta_{(0)}(y))^2}{|x-y|} \d x \d y\\
&\qquad \qquad - \iint \frac{\eta_r(x)^2 |\gamma_0^{1/2}(x,y)|^2}{|x-y|} \d x \d y \\
&\ge \cE^{\rm M} (\eta_r \gamma_0 \eta_r) + \iint_{|y|\le r} \frac{\eta_r(x)^2 \rho_0(x)\rho_0(y)}{|x-y|} \d x \d y \\
&\qquad \qquad - \iint \frac{\eta_r(x)^2 |\gamma_0^{1/2}(x,y)|^2}{|x-y|} \d x \d y \\
&= \cE^{\rm RHF}_r (\eta_r \gamma_0 \eta_r) - X((\eta_r \gamma_0 \eta_r)^{1/2}) - \iint \frac{\eta_r(x)^2 |\gamma_0^{1/2}(x,y)|^2}{|x-y|} \d x \d y \\
&\ge \cE^{\rm RHF}_r (\eta_r \gamma_0 \eta_r) - 3 \Big( \Tr(-\Delta \eta_r \gamma_0 \eta_r) \Big)^{1/2} \Big( \int \eta_r^2 \rho_0 \Big)^{1/2}. 
\end{align*}
In the last inequality we have used \eqref{eq:X-2} twice, once with $\chi=1$ and once with $\chi=\eta_r$.  Similarly, we have
\begin{align*}
&\cE^{\rm M}(\eta_{(0)}\gamma_0 \eta_{(0)}) + \iint \frac{\eta_{(0)}(x)^2 \rho_0(x)\rho_0(y) \eta_-(y)^2}{|x-y|} \d x \d y \\
&\qquad - \iint \frac{\eta_{(0)}(x)^2 |\gamma_0^{1/2}(x,y)|^2}{|x-y|} \d x \d y \\
& \ge \cE^{\rm M}(\eta_{(0)}\gamma_0 \eta_{(0)}) + \iint_{|y|\le (1-\lambda)r} \frac{\eta_{(0)}(x)^2 \rho_0(x)\rho_0(y)}{|x-y|} \d x \d y \\
&\qquad - \iint \frac{\eta_{(0)}(x)^2 |\gamma_0^{1/2}(x,y)|^2}{|x-y|} \d x \d y \\
&= \cE_{(1-\lambda)r}^{\rm RHF}(\eta_{(0)}\gamma_0 \eta_{(0)}) - X ((\eta_{(0)}\gamma_0 \eta_{(0)})^{1/2}) - \iint \frac{\eta_{(0)}(x)^2 |\gamma_0^{1/2}(x,y)|^2}{|x-y|} \d x \d y \\
&\ge \cE_{(1-\lambda)r}^{\rm RHF}(\eta_{(0)}\gamma_0 \eta_{(0)}) - 3 \Big( \Tr(-\Delta \eta_{(0)} \gamma_0 \eta_{(0)}) \Big)^{1/2} \Big( \int \eta_{(0)}^2 \rho_0 \Big)^{1/2}\\
&\ge \Tr\bigl((-(1/2)\Delta - \Phi_{(1-\lambda)r}) \eta_{(0)} \gamma_0 \eta_{(0)} \bigr)-  C \int \eta_{(0)}^2 \rho_0 \,. \numberthis \label{eq:intermediate-region-bound}
\end{align*}
Applying the Lieb--Thirring inequality with $V = \Phi_{(1-\lambda)r} \boldsymbol1_{\supp \eta_{(0)}}$, we obtain
\begin{align*} 
\Tr\bigl((-(1/2)\Delta - \Phi_{(1-\lambda)r}) \eta_{(0)} \gamma_0 \eta_{(0)} \bigr) &\ge \Tr [-(1/2)\Delta - V]_- \ge -C \int V^{5/2}\\
& \ge -C \lambda r^{3} \sup_{\abs{x} \ge (1-\lambda)r}[\Phi_{(1-\lambda)r}(x)]_+^{5/2}.
\end{align*} 
Thus \eqref{eq:intermediate-region-bound} implies that
\begin{align*} 
&\cE^{\rm M}(\eta_{(0)}\gamma_0 \eta_{(0)}) + \iint \frac{\eta_{(0)}(x)^2 \rho_0(x)\rho_0(y) \eta_-(y)^2}{|x-y|} \d x \d y \\
& \qquad - \iint \frac{\eta_{(0)}(x)^2 |\gamma_0^{1/2}(x,y)|^2}{|x-y|} \d x \d y \\
& \ge - C \lambda r^{3 }\sup_{\abs{x} \ge (1-\lambda)r}[\Phi_{(1-\lambda)r}(x)]_+^{5/2} - C \int_{(1-\lambda)r \le |x| \le (1+\lambda)r} \rho_0.
\end{align*}
Putting everyting together, we conclude that
\begin{align*}
\cE^{\rm M}(\gamma_0) &\ge \cE^{\rm M}(\eta_-\gamma_0\eta_-) +  \cE^{\rm RHF}_r (\eta_r \gamma_0 \eta_r) \\
& \quad - C(1+(\lambda r)^{-2}) \int_{(1-\lambda)r\le |x| \le (1+\lambda)r} \rho_0 \\
&\quad -  C \lambda r^{3 }\sup_{\abs{x} \ge (1-\lambda)r}[\Phi_{(1-\lambda)r}(x)]_+^{5/2}  \\
& \quad - 3 \Big( \Tr(-\Delta \eta_r \gamma_0 \eta_r) \Big)^{1/2} \Big( \int \eta_r \rho_0 \Big)^{1/2}.
\end{align*}
This implies the lower bound in \eqref{eq:key-split}. 
\end{proof}

As a by-product of the above proof we obtain 

\begin{lemma} \label{lem:outside-kinetic} For all $r>0$ and all $\lambda\in (0,1/2]$ we have
\begin{align*}
\Tr(-\Delta \eta_r \gamma_0 \eta_r) & \le C (1+ (\lambda r)^{-2}) \int \chi_{(1-\lambda)r}^+\rho_0 \\
& \quad + C\lambda r^3 \sup_{|z|\geq (1-\lambda)r} [\Phi_{(1-\lambda)r}(z)]_+^{5/2} +  C\sup_{|z|\ge r}[|z| \Phi_r(z)]_+^{7/3}.
\end{align*}
\end{lemma}

\begin{proof}
We apply Lemma \ref{lem:split} with $\gamma=0$ and obtain $\cE^{\rm RHF}_r (\eta_r \gamma_0 \eta_r)\leq\mathcal R$. On the other hand, by the kinetic Lieb--Thirring inequality and the fact that the ground state energy in Thomas--Fermi theory is a negative constant times $-Z^{7/3}$, we have
\begin{align*}
\cE^{\rm RHF}_r (\eta_r \gamma_0 \eta_r) \geq & (1/2) \Tr(-\Delta\eta_r \gamma_0 \eta_r) \\
& + C^{-1} \int (\eta_r^2 \rho_0)^{5/3} - \sup_{|z|\geq r} [|z|\Phi_r(z)]_+ \int \frac{\eta_r^2 \rho_0}{|x|} + D(\eta_r^2 \rho_0) \\
=& (1/2) \Tr(-\Delta\eta_r \gamma_0 \eta_r) - C \sup_{|z|\geq r} [|z|\Phi_r(z)]_+^{7/3}.
\end{align*}
Therefore,
$$
\Tr(-\Delta\eta_r \gamma_0 \eta_r) \leq C\mathcal R + C \sup_{|z|\geq r} [|z|\Phi_r(z)]_+^{7/3},
$$
which implies the lemma.
\end{proof}



\section{A collection of useful facts}

\subsection{Semiclassical analysis}

In order to compare M\"uller theory with Tho\-mas--Fermi theory, we use a semiclassical approximation. The following results are taken from \cite[Lemma 8.2]{Solovej-03} (more precisely, we have optimized over $\delta>0$ and changed $V\mapsto 2V$). We put
$$
L_{\rm sc}= (15 \pi^2)^{-1}.
$$

\begin{lemma}\label{lem:semi} For $s>0$, fix a smooth function $g:\R^3\to [0,1]$ such that
$$
\supp g \subset \{|x| \le s\}, \quad \int g^2=1,\quad \int |\nabla g|^2 \le Cs^{-2}. 
$$
(i) For all $V: \R^3\to \R$ such that $[V]_+, [V-V*g^2]_+ \in L^{5/2}$ and for all density matrices $0\le \gamma \le 1$, we have
\begin{align} \label{eq:semi-lower}
\Tr( (-\Delta-V)\gamma) &\ge - L_{\rm sc} \int [V]_+^{5/2} - C s^{-2} \Tr \gamma \nn\\
&\qquad - C \left( \int [V]_+^{5/2}\right)^{3/5}\left( \int [V-V*g^2]_+^{5/2}\right)^{2/5}.
\end{align}

(ii) On the other hand, if $[V]_+\in L^{5/2}\cap L^{3/2}$, then there is a density matrix $\gamma$ such that
$$
\rho_\gamma= \frac{5}{2}L_{\rm sc} [V]_+^{3/2}*g^2
$$
and
\begin{align} \label{eq:semi-upper}
\Tr( -\Delta\gamma) \leq \frac{3}{2}L_{\rm sc}\int [V]_+^{5/2} + C  s^{-2} \int [V]_+^{3/2}.
\end{align}
\end{lemma}

\subsection{Coulomb potential estimate}

The following bound is essentially contained in \cite[Cor. 9.3]{Solovej-03} and appears explicitly in \cite[Lem. 18]{FraNamBos-16} (applied to both $\pm f$). 

\begin{lemma} \label{lem:f*1/|x|} For every $f\in L^{5/3} \cap L^{6/5} (\R^3)$ and $x\in \R^3$, we have
\bq \label{eq:Coulomb-estimate-2}
\left| \int_{|y|<|x|} \frac{f(y)}{|x-y|} \d y \right| \le C \|f\|_{L^{5/3}}^{5/6} (|x|D(f))^{1/12}.
\eq
\end{lemma}


\section{Screened potential estimate}

From now on we always assume that $\gamma_0$ is a minimizer for $E_Z^{\rm M}(N)$ with $N\ge Z \ge 1$. Our main tool to prove the ionization bound is the following 

\begin{lemma}[Screened potential estimate] \label{lem:screened} There are universal constants $C>0,\eps>0,D>0$ such that 
 $$
 | \Phi_{|x|}(x)  - \Phi^{\rm TF}_{|x|}(x)  | \le C |x|^{-4+\eps}, \quad \forall |x|\le D.
 $$
\end{lemma}

We prove Lemma 13 using a bootstrap argument based on two lemmas.
 
\begin{lemma}[Initial step] \label{thm:screened-first} There is a universal constant $C_1>0$  such that 
$$
| \Phi_{|x|}(x)  - \Phi^{\rm TF}_{|x|}(x)  | \le C_1 Z^{49/36-a}|x|^{1/12}, \quad \forall |x|>0,
$$
with $a = 1/198$.
\end{lemma}

\begin{lemma}[Iterative step] \label{thm:screened-it} There are universal constants $C_2, \beta,\delta,\eps >0$ such that, if 
\bq \label{eq:assume-D}
| \Phi_{|x|}(x)  - \Phi^{\rm TF}_{|x|}(x)  | \le \beta |x|^{-4}, \quad \forall |x| \le D
\eq
for some $D\in [Z^{-1/3}, 1]$, then  
\bq \label{eq:assume-D-it}
| \Phi_{|x|}(x)  - \Phi^{\rm TF}_{|x|}(x)  | \le C_2 |x|^{-4+\eps}, \quad \forall D\le |x| \le D^{1-\delta}.
\eq
\end{lemma}

Let us prove Lemma \ref{lem:screened} using Lemmas \ref{thm:screened-first} and \ref{thm:screened-it}.  The proof is identical to \cite[Proof of Lemma 15]{FraNamBos-16}, but is repeated here for the convenience of the reader.

\begin{proof}[Proof of Lemma \ref{lem:screened}]
We use the notations in Lemmas \ref{thm:screened-first} and \ref{thm:screened-it} and set $\sigma=\max\{C_1,C_2\}$. Without loss of generality we may assume that $\beta<\sigma$ and $\eps\le 3a=1/66$. Let us denote
$$D_n=Z^{-\frac{1}{3}(1-\delta)^n}, \quad n=0,1,2,\ldots .$$
From Lemma \ref{thm:screened-first}, we have
$$
|\Phi_{|x|}(x)-\Phi_{|x|}^{\rm TF}(x)|\le C_1 Z^{49/36-a}|x|^{1/12} \le \sigma |x|^{-4+\eps}, \quad \forall |x| \le D_0=Z^{-1/3}.
$$
From Lemma \ref{thm:screened-it}, we deduce by induction that for all $n=0,1,2,...$, if 
$$\sigma (D_n)^\eps \le \beta,$$
then 
$$ |\Phi_{|x|}(x)-\Phi_{|x|}^{\rm TF}(x)|\le \sigma |x|^{-4+\eps}, \quad \forall |x|\le D_{n+1}.$$

Note that $D_n\to 1$ as $n\to\infty$ and that $\sigma>\beta$. Thus, there is a minimal $n_0=0,1,2,\ldots$ such that $\sigma (D_{n_0})^\epsilon>\beta$. If $n_0\ge 1$, then $\sigma (D_{n_0-1})^\epsilon\le\beta$ and therefore by the preceding argument
$$
|\Phi_{|x|}(x)-\Phi_{|x|}^{\rm TF}(x)|\le \sigma |x|^{-4+\eps}, \quad \forall |x|\le D_{n_0} \,.
$$
As we have already shown, the same bound holds for $n_0=0$. Let $D =(\sigma^{-1}\beta)^{1/\epsilon}$, which is a universal constant, and note that by choice of $n_0$ we have $D_{n_0}\geq D$.
\end{proof}

We will prove Lemmas \ref{thm:screened-first} and \ref{thm:screened-it} in the following two sections. 


\section{Initial step}

In this section we prove Lemma \ref{thm:screened-first}. Recall that we always assume $Z \ge 1$. For simplicity we write
$$
\cE^{\rm RHF}(\gamma) = \cE^{\rm RHF}_{r=0}(\gamma) = \Tr (-\Delta \gamma) - \int_{\R^3} \frac{Z\rho_\gamma(x)}{|x|} \d x + D(\rho_\gamma).
$$

\begin{proof}[Proof of Lemma \ref{thm:screened-first}] 
The strategy is to bound $\cE^{\rm M} (\gamma_0)$ from above and below using the semi-classical estimates from Lemma \ref{lem:semi}. The main term in both bounds is $\cE^{\rm TF}(\rho^{\rm TF})$, but in the lower bound we will get an additional term $D(\rho_0-\rho^{\rm TF})$. The error terms in the upper and lower bounds will then give an upper bound on $D(\rho_0-\rho^{\rm TF})$ which will imply the lemma.

\textbf{Upper bound.} We shall show that
\begin{equation}
\label{eq:initialupper}
\cE^{\rm M}(\gamma_0) \leq \cE^{\rm TF}(\rho^{\rm TF}) + C Z^{11/5}.
\end{equation}
Indeed, since
$$
N\mapsto \inf\{ \cE^{\rm M}_Z(\gamma) |\ 0\leq\gamma\leq 1,\ \Tr\gamma=N\}
$$
is non-increasing (indeed, it is even non-increasing when $N/16$ is added, see \cite[Proposition 4]{FraLieSieSei-07} and note the different convention for the kinetic energy) and since the contribution of the exchange term to the energy is non-positive, we have
\begin{align}\label{eq:initialupper1}
\cE^{\rm M}(\gamma_0) & \leq \inf\{ \cE^{\rm M}_Z(\gamma) |\ 0\leq\gamma\leq 1,\ \Tr\gamma\leq N\} \notag \\
& \leq \inf\{ \cE^{\rm RHF}(\gamma) |\ 0\leq\gamma\leq 1,\ \Tr\gamma\leq N \}.
\end{align}
Thus, \eqref{eq:initialupper} follows from a well-known bound on the ground state energy in reduced Hartree--Fock theory (essentially in \cite[Proof of Theorem 5.1]{Lieb-81b}), but we include a proof for the sake of completeness. We introduce the Thomas--Fermi potential
$$
\varphi^{\rm TF}(x)=\frac{Z}{|x|} - \rho^{\rm TF}* |x|^{-1}
$$
and apply Lemma \ref{lem:semi} (ii) with $V=\varphi^{\rm TF}$ and a spherically symmetric $g$ to obtain a density matrix $\gamma'$. Because of the Thomas--Fermi equation we have
$$
\rho_{\gamma'} = \frac{5}{2}L_{\rm sc} \left( \varphi^{\rm TF}\right)^{3/2}*g^2 = \rho^{\rm TF}* g^2.
$$
Since
$$
\Tr\gamma' = \int \rho_{\gamma'} = \int \rho^{\rm TF} =Z \leq N,
$$
we obtain
$$
\inf\{ \cE^{\rm RHF}(\gamma) |\ 0\leq\gamma\leq 1,\ \Tr\gamma\leq N\} \leq \cE^{\rm RHF}(\gamma').
$$
By the semiclassical estimate from Lemma~\ref{lem:semi} (ii)
\begin{align*}
\cE^{\rm RHF} ( \gamma') & \le 
        \frac{3}{2}L_{\rm sc}\int [V]_+^{5/2} + C  s^{-2} \int [V]_+^{3/2} - \int \frac{Z}{|x|}  \left(\rho^{\rm TF}*g^2\right)  \\
        & \qquad + D(\rho^{\rm TF}* g^2)  \\ 
        &\le c^{\rm TF} \int [\rho^{\rm TF}]_+^{5/3} - \int \frac{Z}{|x|} \rho^{\rm TF} + D(\rho^{\rm TF}) \\
        & \qquad  + C  s^{-2} \int \rho^{\rm TF} + \int \left(\frac{Z}{|x|}-\frac{Z}{|x|}*g^2\right) \rho^{\rm TF} \\
        & = \cE^{\rm TF}(\rho^{\rm TF})   + C  s^{-2} \int \rho^{\rm TF} + \int \left(\frac{Z}{|x|}-\frac{Z}{|x|}*g^2\right) \rho^{\rm TF},
\end{align*}
where we have used the convexity of $D$ in the second inequality. By Newton's theorem, $|x|^{-1} - |x|^{-1}*g^2$ is non-negative, bounded by $|x|^{-1}$ and vanishes when $|x|>s$. Moreover, bounding $\varphi^{\rm TF}$ in the Thomas--Fermi equation from above by $Z|x|^{-1}$, we find
$$
\rho^{\rm TF}(x) \leq \left( \frac{3}{5 c^{\rm TF}} \right)^{3/2} Z^{3/2} |x|^{-3/2}.
$$
These facts yield
$$
\int \left(\frac{Z}{|x|}-\frac{Z}{|x|}*g^2\right) \rho^{\rm TF} \leq C Z^{5/2} s^{1/2}.
$$
Thus, after optimization in $s$,
$$
\cE^{\rm RHF} ( \gamma') \leq \cE^{\rm TF}(\rho^{\rm TF}) + C Z^{11/5}.
$$
Combining this with \eqref{eq:initialupper1} we obtain \eqref{eq:initialupper}.

\textbf{Lower bound.}
We now show that
\begin{equation}
\label{eq:initiallower}
\cE^{\rm M} (\gamma_0) \ge \cE^{\rm TF}(\rho^{\rm TF}) + D(\rho_0-\rho^{\rm TF}) - CZ^{25/11}.
\end{equation}
With the Thomas--Fermi potential introduced above we can write
$$
\cE^{\rm M} (\gamma_0) = \Tr ((-\Delta-\varphi^{\rm TF})\gamma_0) + D(\rho_0-\rho^{\rm TF}) - D(\rho^{\rm TF}) - X(\gamma_0^{1/2}).
$$
According to \eqref{eq:apriorieximpr} we can bound the exchange term by
$$X(\gamma_0^{1/2}) \le C Z^{5/3}.$$
Next, from the semiclassical estimate \eqref{eq:semi-lower} we have
\begin{align*} 
\Tr( (-\Delta-\varphi^{\rm TF})\gamma_0) &\ge - L_{\rm sc} \int [\varphi^{\rm TF}]_+^{5/2} - C s^{-2}\Tr \gamma_0 \nn\\
&\qquad - C \left( \int [\varphi^{\rm TF}]_+^{5/2}\right)^{3/5}\left( \int [\varphi^{\rm TF}-\varphi^{\rm TF}*g^2]_+^{5/2}\right)^{2/5}.
\end{align*} 
According to \eqref{eq:roughbound} we can bound $\Tr \gamma_0=N\le CZ$. Moreover, by scaling,
$$
\int |\varphi^{\rm TF}|^{5/2} = C \int (\rho^{\rm TF})^{5/3} \le C Z^{7/3} 
$$
and, as explained in \cite[end of page 554]{Solovej-03},
$$ \int |\varphi^{\rm TF}- \varphi^{\rm TF}*g^2|^{5/2} \le CZ^{5/2} s^{1/2}.$$
Thus,
\begin{align*} 
\Tr( (-\Delta-\varphi^{\rm TF})\gamma_0) &\ge - L_{\rm sc} \int [\varphi^{\rm TF}]_+^{5/2} - Cs^{-2}Z - C Z^{12/5} s^{1/5}.
\end{align*} 
Optimizing over $s>0$ we get
\begin{align*} 
\Tr( (-\Delta-\varphi^{\rm TF})\gamma_0) &\ge - L_{\rm sc} \int [\varphi^{\rm TF}]_+^{5/2} - CZ^{25/11}.
\end{align*} 
Note that from the Thomas--Fermi equation, we have
\begin{equation}
\label{eq:initiallowercomp}
- L_{\rm sc} \int [\varphi^{\rm TF}]_+^{5/2} - D(\rho^{\rm TF}) = \cE^{\rm TF}(\rho^{\rm TF}),
\end{equation}
which proves \eqref{eq:initiallower}.

\textbf{Conclusion.} Combining \eqref{eq:initialupper} and \eqref{eq:initiallower} we infer that
$$
D(\rho_0-\rho^{\rm TF})  \le CZ^{25/11}.
$$
From the Coulomb estimate \eqref{eq:Coulomb-estimate-2} with $f= \rho_0-\rho^{\rm TF}$ and the kinetic estimates
$$\int \rho_0^{5/3} \le CZ^{7/3}, \quad \int (\rho^{\rm TF})^{5/3} \le CZ^{7/3},$$
(the first one follows from \eqref{eq:apriorikinrepimpr} and the second one follows simply by scaling) 
we find that for all $|x|>0$,
\begin{align*}
|\Phi_{|x|}(x)-\Phi_{|x|}^{\rm TF}(x)| &=\left| \int_{|y|<|x|} \frac{\rho_0(y) - \rho^{\rm TF}(y)}{|x-y|} \d y \right| \\
&\le C \|\rho_0-\rho^{\rm TF}\|_{L^{5/3}}^{5/6} (|x|D(\rho_0-\rho^{\rm TF}))^{1/12} \\
&\le CZ^{179/132}|x|^{1/12}.
\end{align*}
Since $179/132= 49/36- 1/198$, this is the desired bound.
\end{proof}


\section{Iterative step}

In this section we will prove Lemma \ref{thm:screened-it}. We split the proof into five steps.

\noindent
{\bf Step 1.} We collect some easy consequences of \eqref{eq:assume-D}. 

\begin{lemma} \label{lem:screened-easy-bounds} 
Assume that \eqref{eq:assume-D} holds true for some $\beta, D\in (0,1]$. Then for all $r\in (0, D]$, we have
\begin{align}\label{eq:int-rho-1}
\left| \int_{|x|<r} (\rho_0 - \rho^{\rm TF}) \right| &\le \beta r^{-3},\\
\label{eq:int-rho-2}
\sup_{|x| \ge r} |x| |\Phi_{r}(x)| &\le C r^{-3},\\
\int_{|x|>r} \rho_0 &\le C r^{-3},\label{eq:int-rho-4}\\
\int_{|x|>r} \rho_0^{5/3} &\le C r^{-7},\label{eq:int-rho-3}\\
\Tr(-\Delta \eta_r \gamma_0 \eta_r)  &\le C (r^{-7}+\lambda^{-2}r^{-5}), \quad \forall \lambda \in (0,1/2].\label{eq:int-rho-2.5}
\end{align}
\end{lemma}

We emphasize that, while \eqref{eq:int-rho-1} and \eqref{eq:int-rho-2} are straightforward consequences of \eqref{eq:assume-D}, the bounds \eqref{eq:int-rho-2.5}, \eqref{eq:int-rho-3} and \eqref{eq:int-rho-4} rely on the outside $L^1$ bound from Lemma \ref{lem:L1-bound} and the outside kinetic energy bound from Lemma \ref{lem:outside-kinetic}. Recall that the smooth cut-off function $\eta_r$ is defined in \eqref{eq:def-eta-r} and depends on the parameter $\lambda$. 

\begin{proof}
The proofs of \eqref{eq:int-rho-1} and \eqref{eq:int-rho-2} can be carried out similarly to \cite[Proof of Lemma 20]{FraNamBos-16}. Note that, in contrast to the corresponding statement in \cite[Lemma 20]{FraNamBos-16}, in \eqref{eq:int-rho-2} we claim a bound on $|\Phi_r(x)|$. This, however, follows in the same way by applying \cite[Lemma 19]{FraNamBos-16} both to $\Phi_r$ and to $-\Phi_r$.

Note that there is an alternative, simpler proof of \eqref{eq:int-rho-1} and \eqref{eq:int-rho-2} based on the spherical symmetry of $\rho_0$, which follows from the convexity of the M\"uller functional. It implies that $\Phi_r(x)-\Phi_r^{\rm TF}(x) = |x|^{-1} \int_{|y|<r} (\rho_0 - \rho^{\rm TF}) \d y$ for $|x|\geq r$, so that \eqref{eq:int-rho-1} follows immediately from \eqref{eq:assume-D} and \eqref{eq:int-rho-2} follows from \eqref{eq:int-rho-1} and a corresponding bound for $\Phi_r^{\rm TF}(x)$.

Now we prove \eqref{eq:int-rho-4} and \eqref{eq:int-rho-2.5}. By \eqref{eq:int-rho-1} and the bound $\rho^{\rm TF}\le C |x|^{-6}$, we have
\begin{align} \label{eq:int-rho-r-r/2}
\int_{r/3<|x|<r} \rho_0 &=  \int_{|x|<r} (\rho_0 - \rho^{\rm TF}) - \int_{|x|<r/3} (\rho_0 - \rho^{\rm TF}) + \int_{r/3<|x|<r} \rho^{\rm TF} \nn\\
& \le \beta r^{-3} +  \beta (r/3)^{-3}  + Cr^{-3} \le C r^{-3}.
\end{align}
Inserting this and the bound \eqref{eq:int-rho-2} into the bound from Lemma~\ref{lem:outside-kinetic}, we obtain
\begin{align}
\Tr(-\Delta \eta_r \gamma_0 \eta_r) & \le C (1+ (\lambda r)^{-2}) \int \chi_{(1-\lambda)r}^+\rho_0 \nn\\
& \quad + C\lambda r^3 \sup_{|z|\geq (1-\lambda)r} [\Phi_{(1-\lambda)r}(z)]_+^{5/2} +  C\sup_{|z|\ge r}[|z| \Phi_r(z)]_+^{7/3} \nn\\
&\le C \Big(\lambda^{-2} r^{-2} \int \chi_r^+ \rho_0 + \lambda^{-2}r^{-5}  + r^{-7} \Big). \label{eq:Tr-eta-int-chir}
\end{align}
In \eqref{eq:Tr-eta-int-chir}, replacing $r$ by $r/3$ and using again \eqref{eq:int-rho-r-r/2} we get
\begin{align}
\Tr(-\Delta \eta_{r/3} \gamma_0 \eta_{r/3})  \le C \Big(\lambda^{-2} r^{-2}  \int \chi_r^+ \rho_0+ \lambda^{-2}r^{-5}  + r^{-7} \Big). \label{eq:Tr-eta-int-chir-a}
\end{align}
From the exterior bound from Lemma~\ref{lem:L1-bound}, replacing $r$ by $r/3$ and choosing $s=r$ (this choice is not optimal but sufficient), we find that
\begin{align*} 
 \int \chi^+_{r/3}\rho_0  &\le C  \int_{r/3<|x|<(1+\lambda)^2 r/3} \rho_0 + C \Big(\sup_{|z|\ge r/3} [|z|\Phi_{r/3}(z)]_+  + \lambda^{-2}r^{-1}\Big)  \\
 &\quad +C  \Big( r^2 \Tr (-\Delta \eta_{r/3} \gamma_0 \eta_{r/3}) \Big)^{3/5} + C \Big(r^2 \Tr (-\Delta \eta_{r/3} \gamma_0 \eta_{r/3}) \Big)^{1/3}.
\end{align*} 
Inserting \eqref{eq:int-rho-2}, \eqref{eq:int-rho-r-r/2} and \eqref{eq:Tr-eta-int-chir} into the latter estimate leads to
\begin{align*} 
 \int \chi_r^+ \rho_0 \le \int \chi^+_{r/3}\rho_0  &\le C (r^{-3}  + \lambda^{-2} r^{-1})  \\
 &\quad + C \Big(\lambda^{-2}  \int \chi_r^+ \rho_0 + \lambda^{-2}r^{-3}  + r^{-5} \Big)^{3/5} \\
&\quad +  C \Big(\lambda^{-2}  \int \chi_r^+ \rho_0 + \lambda^{-2}r^{-3}  + r^{-5} \Big)^{1/3}
\end{align*}
which implies \eqref{eq:int-rho-4} immediately (we can choose $\lambda=1/2$ on the right side). Inserting \eqref{eq:int-rho-4} into \eqref{eq:Tr-eta-int-chir} we obtain \eqref{eq:int-rho-2.5}. 

Finally, from \eqref{eq:int-rho-2.5} and the kinetic Lieb--Thirring inequality, we have
\[
\int_{|x|>r} \rho_0^{5/3} \le \int (\eta_{r/3}^2 \rho_0)^{5/3} \le C \Tr(-\Delta \eta_{r/3} \gamma_0 \eta_{r/3}) \le C (r^{-7}+\lambda^{-2}r^{-5}).
\]
which implies \eqref{eq:int-rho-3} (we can choose $\lambda=1/2$ on the right side). 
\end{proof}

\noindent
{\bf Step 2.} \noindent 
Now we introduce the exterior Thomas--Fermi energy functional
$$
\cE_r^{\rm TF}(\rho)= c^{\rm TF}\int \rho^{5/3} - \int V_r \rho + D(\rho), \quad V_r(x)=\chi_r^+ \Phi_r(x)
$$
with $c_{\rm TF}$ from \eqref{eq:ctf}.

\begin{lemma}\label{exteriortf}
 The  functional $\cE_r^{\rm TF}(\rho)$ has a unique minimizer $\rho_r^{\rm TF}$ over
$$
0\le \rho \in L^{5/3}(\R^3) \cap L^1(\R^3), \quad \int \rho \le \int \chi_r^+ \rho_0.
$$
This minimizer is supported on $\{|x|\ge r\}$ and satisfies the Thomas--Fermi equation
$$
\frac{5c^{\rm TF}}{3} \rho_r^{\rm TF}(x)^{2/3} = [\varphi_r^{\rm TF}(x)-\mu_r^{\rm TF}]_+
$$
with $\varphi_r^{\rm TF}(x)= V_r(x)-\rho_r^{\rm TF}*|x|^{-1}$ and a constant $\mu_r^{\rm TF} \ge 0$. Moreover, if \eqref{eq:assume-D} holds true for some $\beta, D\in (0,1]$, then 
\bq \label{eq:rho-r-TF-5/3}
\int (\rho_r^{\rm TF})^{5/3} \le Cr^{-7}, \quad \forall r\in (0,D].
\eq
\end{lemma}

The proof of this lemma is identical to that of \cite[Lemma 21]{FraNamBos-16} and is omitted.

\bigskip

\noindent 
{\bf Step 3.} Now we compare $\rho_r^{\rm TF}$ with $\chi_r^+\rho^{\rm TF}$. 

\begin{lemma} \label{lem:varphirTF-varphiTF} We can choose a universal constant $\beta>0$ small enough such that, if \eqref{eq:assume-D} holds true for some $D\in [Z^{-1/3},1]$, then $\mu_r^{\rm TF}=0$ and 
\begin{align*}
\left| \varphi_r^{\rm TF}(x) - \varphi^{\rm TF}(x) \right| \le C (r/|x|)^{\zeta}|x|^{-4}, \\
\left| \rho_r^{\rm TF}(x) - \rho^{\rm TF}(x) \right| \le C (r/|x|)^{\zeta}|x|^{-6}
\end{align*}
for all $r\in [Z^{-1/3},D]$ and for all $|x| \ge r$. Here $\zeta=(\sqrt{73}-7)/2\approx 0.77$.
\end{lemma}

The proof of this lemma is identical to that of \cite[Lemma 22]{FraNamBos-16} and is omitted.

\bigskip

\noindent
{\bf Step 4.} In this step, we compare $\rho_r^{\rm TF}$ with $\chi_r^+\rho_0$.

\begin{lemma} \label{lem:DrhorTF-rho0}Let $\beta>0$ be as in Lemma \ref{lem:varphirTF-varphiTF}. Assume that \eqref{eq:assume-D} holds true for some $D\in [Z^{-1/3},1]$. Then 
$$
D(\chi_r^+\rho_0- \rho_r^{\rm TF})\le C r^{-7+b}, \quad \forall r\in [Z^{-1/3},D],
$$
where $b = 1/3$.
\end{lemma}

\begin{proof}
The proof of this lemma is analogous, but somewhat more involved than the proof of Lemma \ref{thm:screened-first}. The strategy is to bound $\cE_r^{\rm RHF} (\eta_r \gamma_0 \eta_r)$ from above and below using the semi-classical estimates from Lemma \ref{lem:semi}. The main term in both bounds is $\cE^{\rm TF}_r(\rho_r^{\rm TF})$, but in the lower bound we will get an additional term $D(\chi_r^+\rho_0-\rho_r^{\rm TF})$. The error terms in the upper and lower bounds will then give the desired bound in the lemma.

{\bf Upper bound.} 
We shall prove that
\begin{equation}
\label{eq:iterativeupper}
\cE_r^{\rm RHF} (\eta_r \gamma_0 \eta_r) \leq \cE^{\rm TF}_r(\rho_r^{\rm TF}) + C r^{-7} \left( r^{2/3} + \lambda^{-2} r^2 + \lambda \right).
\end{equation}
We use Lemma~\ref{lem:semi} (ii) with $V_r' \equiv \chi_{r+s}^+\varphi_r^{\rm TF} $, $s \le r$ to be chosen later and $g$ spherically symmetric to obtain a density matrix $\gamma_r$ as in the statement. Since $\mu_r^{\rm TF} = 0$ by Lemma \ref{lem:varphirTF-varphiTF}, we deduce from the Thomas--Fermi equation in Lemma \ref{exteriortf} that
\[
\rho_{\gamma_r} = \frac{5}{2}L_{\rm sc} \left( \chi_{r+s}^+ \left(\varphi_r^{\rm TF}\right)^{3/2}\right)*g^2 = \left(\chi_{r+s}^+ \rho_r^{\rm TF}\right)* g^2.
\]
Note that $\rho_{\gamma_r}$ is supported in $\{|x|\ge r\}$ and
$$
\Tr\gamma_r = \int \rho_{\gamma_r} = \int \chi_{r+s}^+ \rho_r^{\rm TF} \leq \int \rho_r^{\rm TF} \leq \int \chi_r^+\rho_0.
$$
Thus, we may apply Lemma~\ref{lem:split} and obtain
\begin{align}
\label{eq:iterativeupper1}
\cE_r^{\rm RHF} (\eta_r \gamma_0 \eta_r) & \le \cE_r^{\rm RHF} (\gamma_r) + \cR. 
\end{align}
By the semiclassical estimate from Lemma~\ref{lem:semi} (ii)
\begin{align*}
\cE_r^{\rm RHF} ( \gamma_r) & \le 
        \frac{3}{2}L_{\rm sc}\int [V_r']_+^{5/2} + C  s^{-2} \int [V_r']_+^{3/2} - \int \Phi_r  \left((\chi_{r+s}^+\rho_r^{\rm TF})*g^2\right)  \\
        & \qquad + D(\rho_r^{\rm TF}* g^2)  \\ 
        &\le \frac{3}{2}L_{\rm sc}\int [\varphi_r^{\rm TF}]_+^{5/2} - \int \Phi_r \rho_r^{\rm TF} + D(\rho_r^{\rm TF}) \\
        & \qquad  + C  s^{-2} \int \rho_r^{\rm TF} + \int (\Phi_r-\Phi_r*g^2) \chi_{r+s}^+\rho_r^{\rm TF} 
            + \int_{r \le \abs{x} \le r+s} \Phi_r \rho_r^{\rm TF} \\
        & = \cE^{\rm TF}_r(\rho_r^{\rm TF})   + C  s^{-2} \int \rho_r^{\rm TF} + \int_{r \le \abs{x} \le r+s} \Phi_r \rho_r^{\rm TF},
\end{align*}
where we have used the convexity of $D$ in the second inequality.
The equality in the last line holds, since $\Phi_r(x)$ is harmonic when $\abs{x} \ge r$ and $g$ is chosen spherically symmetric.  

According to \eqref{eq:int-rho-4} in Lemma~\ref{lem:screened-easy-bounds} we have
$$
\int \rho_r^{\rm TF} \leq \int \chi_r^+ \rho_0 \leq C r^{-3}.
$$
We now use the fact that $\rho_r^{\rm TF}(x) \leq C |x|^{-6}$ for all $|x|\geq r$, which follows from Lemma \ref{lem:varphirTF-varphiTF} because of the corresponding bound for $\rho^{\rm TF}$. (In fact, the claimed upper bound does not need the full strength of Lemma \ref{lem:varphirTF-varphiTF}, so as part of the proof one shows that $\varphi_r^{\rm TF}(x)\leq C r^{-4}$ for $|x|\geq r$ and therefore the claimed bound follows from the Thomas--Fermi equation for $\rho_r^{\rm TF}$.) Thus, using \eqref{eq:int-rho-2} in Lemma \ref{lem:screened-easy-bounds},
$$
\int_{r \le \abs{x} \le r+s} \Phi_r \rho_r^{\rm TF} \leq C r^{-3} \int_{r \le \abs{x} \le r+s} |x|^{-1} \rho_r^{\rm TF} \leq C r^{-8} s.
$$
Optimizing over $s$ (which leads to $s \sim r^{5/3}$) we obtain
\begin{equation}
\label{eq:iterativeupper2}
\cE_r^{\rm RHF} ( \gamma_r)\leq \cE^{\rm TF}_r(\rho_r^{\rm TF}) + C r^{-7+2/3}.
\end{equation}

Finally, we estimate $\mathcal R$ using Lemma~\ref{lem:screened-easy-bounds} and obtain
\begin{align*}
\cR &\le C  (1+(\lambda r)^{-2}) r^{-3} + C \lambda r^3 (r^{-4})^{5/2} + C (r^{-7}+\lambda^{-2}r^{-5})^{1/2} (r^{-3})^{1/2} \\
& \leq  C \left( \lambda^{-2} r^{-5} + \lambda r^{-7} \right) .
\end{align*}
Combining this with \eqref{eq:iterativeupper1} and \eqref{eq:iterativeupper2} we obtain the claimed upper bound \eqref{eq:iterativeupper}.

{\bf Lower bound.}
We shall prove that
\begin{equation}
\label{eq:iterativelower}
\cE_r^{\rm RHF} (\eta_r \gamma_0 \eta_r)  \geq \cE^{\rm TF}_r(\rho_r^{\rm TF}) + D(\eta_r^2 \rho_0- \rho_r^{\rm TF}) - C r^{-7+ 1/3}
\end{equation}

We use Lemma~\ref{lem:semi} (i) in a way similar to the proof of Lemma~\ref{thm:screened-first} to obtain
\begin{align*}
 \cE_r^{\rm RHF} (\eta_r \gamma_0 \eta_r) 
    &= \Tr ((-   \Delta - \varphi_r^{\rm TF})\eta_r \gamma_0 \eta_r) + D(\eta_r^2 \rho_0- \rho_r^{\rm TF}) - D(\rho_r^{\rm TF}) \\
    & \ge - L_{\rm sc } \int[\varphi_r^{\rm TF}]_+^{5/2} - C s^{-2} \int \eta_r^2 \rho_0 \\
    & \qquad - C \left(  \int   [\varphi_r^{\rm TF}]_+^{5/2}  \right)^{3/5}
	  \left(\int [\varphi_r^{\rm TF}-\varphi_r^{\rm TF}*g^2]_+^{5/2} \right)^{2/5} \\
    & \qquad + D(\eta_r^2 \rho_0- \rho_r^{\rm TF}) - D(\rho_r^{\rm TF})\\
    & =  \cE^{\rm TF}_r(\rho_r^{\rm TF})  + D(\eta_r^2 \rho_0- \rho_r^{TF})
    - C s^{-2} \int \eta_r^2 \rho_0 \\
    & \qquad - C \left(  \int   [\varphi_r^{\rm TF}]_+^{5/2}  \right)^{3/5}
	  \left(\int [\varphi_r^{\rm TF}-\varphi_r^{\rm TF}*g^2]_+^{5/2} \right)^{2/5}.
\end{align*}
The last identity used the Thomas--Fermi equations similarly as in \eqref{eq:initiallowercomp}.

In order to control the remainder terms we note that by Lemmas~\ref{lem:screened-easy-bounds} and \ref{exteriortf} we have 
\begin{align*}
 \int \eta_r^2 \rho_0 & \le C r^{-3} \,,
 \qquad
 \int   [\varphi_r^{\rm TF}]_+^{5/2} = C \int \left(\rho_r^{\rm TF}\right)^{5/3} \le C r^{-7} \,.
\end{align*}
In order to bound the convolution term we use, as in the proof of Lemma \ref{thm:screened-first} the fact that $|x|^{-1} - |x|^{-1}* g^2\geq 0$, and therefore also $\rho_r^{\rm TF}* (|x|^{-1} - |x|^{-1}* g^2)\geq 0$. Since $\varphi_r^{\rm TF} = \chi_r^+ \Phi_r - \rho_r^{\rm TF} * |x|^{-1}$, we conclude that
$$
\varphi_r^{\rm TF}-\varphi_r^{\rm TF}*g^2 \leq \chi_r^+ \Phi_r - (\chi_r^+ \Phi_r) * g^2.
$$
Since $\Phi_r$ is harmonic outside a ball of radius $r$ and $g$ is spherically symmetric, $\chi_r^+ \Phi_r - (\chi_r^+ \Phi_r) * g^2$ is supported in $\{r-s\leq |x|\leq r+s\}$ and, by Lemma \ref{lem:screened-easy-bounds}, its absolute value is bounded by $C r^{-4}$. Thus,
$$
[\varphi_r^{\rm TF}-\varphi_r^{\rm TF}*g^2]_+ \leq C r^{-4} \1\big(r-s \le \abs \dotv \le r+s \big)
$$
and therefore
$$
\int [\varphi_r^{\rm TF}-\varphi_r^{\rm TF}*g^2]_+^{5/2} \leq C r^{-8} s
$$

To summarize, we have shown that 
$$
\cE_r^{\rm RHF} (\eta_r \gamma_0 \eta_r) 
\ge  \cE^{\rm TF}_r(\rho_r^{\rm TF})  + D(\eta_r^2 \rho_0- \rho_r^{\rm TF})
    - C \left( s^{-2} r^{-3} + r^{-37/5} s^{2/5} \right).
$$ 
Optimizing over $s$ (which leads to $s\sim r^{11/6}$) we obtain \eqref{eq:iterativelower}.

\textbf{Conclusion.}
Combining \eqref{eq:iterativeupper} and \eqref{eq:iterativelower} we infer that
\[
D(\eta_r^2 \rho_0- \rho_r^{\rm TF}) \le C r^{-7} \left( r^{1/3} +  \lambda^{-2} r^2 + \lambda \right).
\]
The next step is to replace $\eta_r^2$ by $\chi_r^+$. By using the Hardy--Littewood--Soloblev inequality and \eqref{eq:int-rho-3}, we get 
\begin{align*}
D(\chi_r^+ \rho_0 - \eta_r^2 \rho_0) &\le D(\1\big((1+\lambda)r \ge |x| \ge r\big) \rho_0) \\
&\le C \| \1\big((1+\lambda)r \ge |x| \ge r\big) \rho_0\|_{L^{6/5}}^2\\
&\le C \left( \int \chi_r^+ \rho_0^{5/3}\right)^{6/5} \left( \int_{(1+\lambda)r \ge |x| \ge r} \right)^{7/15} \\
&\le C (r^{-7})^{6/5} (\lambda r^3)^{7/15} = C\lambda^{7/15} r^{-7}.
\end{align*}
Therefore, 
\begin{align*}
D(\chi_r^+\rho_0 - \rho_r^{\rm TF}) &\le 2 D(\chi_r^+\rho_0 - \eta_r^2 \rho_0) + 2D(\eta_r^2 \rho_0 - \rho_r^{\rm TF}) \\
&\le C r^{-7} \left(\lambda^{7/15} + r^{1/3} + \lambda^{-2}r^2 \right).
\end{align*}
This bound is valid for all $\lambda \in (0,1/2]$ and by optimizing over $\lambda$ (which leads to $\lambda \sim r^{30/37}$) we obtain
\begin{align*}
D(\chi_r^+\rho_0 - \rho_r^{\rm TF}) \le C r^{-7+1/3}.
\end{align*}
This is the desired estimate.
\end{proof}

\noindent
{\bf Step 5.} Now we are ready to conclude. The argument is similar to that in \cite[Lemma 17]{FraNamBos-16}, but for the convenience of the reader we provide the details.

\begin{proof}[Proof of Lemma \ref{thm:screened-it}]
Let $r\in [Z^{-1/3},D]$ and $|x|\ge r$. As in \cite[Eq. (97)]{Solovej-03}, we can decompose 
\begin{align*}
\Phi_{|x|}(x)-\Phi_{|x|}^{\rm TF}(x) &= \varphi_r^{\rm TF}(x) -\varphi^{\rm TF}(x)+ \int_{|y|>|x|} \frac{\rho_r^{\rm TF}(y)-\rho^{\rm TF}(y)}{|x-y|} \d y,\\
&\qquad + \int_{|y|<|x|}  \frac{\rho_r^{\rm TF}(y)-(\chi_r^+\rho_0)(y)}{|x-y|} \d y.
\end{align*}
By Lemma \ref{lem:varphirTF-varphiTF}, we have
\begin{align*}
\left| \varphi_r^{\rm TF}(x) - \varphi^{\rm TF}(x) \right| \le C (r/|x|)^{\zeta}|x|^{-4}
\end{align*}
and
\begin{align*}
\int_{|y|>|x|} \frac{|\rho_r^{\rm TF}(y)-\rho^{\rm TF}(y)|}{|x-y|} \d y \le C \int_{|y|>|x|} \frac{(r/|y|)^{\zeta}|y|^{-6}}{|x-y|} \d y \le C (r/|x|)^{\zeta}|x|^{-4}.
\end{align*}
Moreover, from \eqref{eq:Coulomb-estimate-2}, \eqref{eq:int-rho-3}, \eqref{eq:rho-r-TF-5/3} and Lemma \ref{lem:DrhorTF-rho0}, we get
\begin{align*}
\left| \int_{|y|<|x|}  \frac{\rho_r^{\rm TF}(y)-(\chi_r^+\rho_0)(y)}{|x-y|} \d y \right| &\le C \| \rho_r^{\rm TF}-\chi_r^+\rho_0\|_{L^{5/3}}^{5/6} \Big(|x| D(\rho_r^{\rm TF}-\chi_r^+\rho_0)\Big)^{1/12} \\
&\le C (r^{-7})^{1/2} (|x| r^{-7+b} )^{1/12} \\
& =   C |x|^{-4+b/12} (|x|/r)^{4+1/12-b/12}   ,
\end{align*}
Thus in summary, for all $r\in [Z^{-1/3},D]$ and $|x|\ge r$, we have
\bq \label{eq:bootstrap-total}|\Phi_{|x|}(x)-\Phi_{|x|}^{\rm TF}(x)|\le C (r/|x|)^\zeta |x|^{-4}+ C (|x|/r)^{5} |x|^{-4+b/12}.\eq

Now let us conclude using \eqref{eq:bootstrap-total}. We choose a constant $\delta\in(0,1)$ sufficiently small such that
\begin{equation}
\label{eq:delta1}
\frac{1+\delta}{1-\delta}  \Big( \frac{49}{36} -a \Big) < \frac{49}{36}
\end{equation}
and
\begin{equation}
\label{eq:delta2}
\frac{b}{12} -\frac{10 \delta}{1-\delta}>0.
\end{equation}
We recall that $a$ and $b$ are the constants from Lemmas \ref{thm:screened-first} and \ref{lem:DrhorTF-rho0}, respectively. We distinguish two cases. 

{\bf Case 1:} $D^{1+\delta}\le Z^{-1/3}$. In this case, we simply use the initial step. Indeed, for all
$$|x|\le D^{1-\delta}\le (Z^{-1/3})^{(1-\delta)/(1+\delta)},$$
by Lemma \ref{thm:screened-first} we have
\bq \label{eq:bootstrap-total-1a}
|\Phi_{|x|}(x)-\Phi_{|x|}^{\rm TF}(x)|\le C_1 Z^{49/36-a} |x|^{1/12} \le C_1 |x|^{1/12- 3 \frac{1+\delta}{1-\delta} (49/36-a)}.
\eq
Note that
$$
\frac{1}{12} - 3 \times \frac{49}{36} = -4
$$
Therefore, \eqref{eq:delta1} implies that
$$
\frac{1}{12} - \frac{3(1+\delta)}{1-\delta}  \Big( \frac{49}{36} -a \Big)> -4.
$$

{\bf Case 2:} $D^{1+\delta}\ge Z^{-1/3}$. In this case, we use \eqref{eq:bootstrap-total} with $r=D^{1+\delta}$. For all $D \le |x| \le D^{1-\delta}$ we have
$$ |x|^{2\delta/(1-\delta)} \le r/|x| \le |x|^\delta.$$
Therefore, \eqref{eq:bootstrap-total} implies that 
\bq \label{eq:bootstrap-total-1b} |\Phi_{|x|}(x)-\Phi_{|x|}^{\rm TF}(x)|\le C |x|^{-4+\zeta \delta} + C|x|^{-4+b/12 -10 \delta/(1-\delta)}.
\eq
Both exponents of $|x|$ are strictly greater than $-4$ according to \eqref{eq:delta2}.

In summary, from \eqref{eq:bootstrap-total-1a} and \eqref{eq:bootstrap-total-1b}, we conclude that in both cases,
$$ |\Phi_{|x|}(x)-\Phi_{|x|}^{\rm TF}(x)|\le C |x|^{-4+\eps}, \quad \forall D \le |x| \le D^{1-\delta}.$$
with  
$$\eps:=\min\Big\{\frac{1}{12} - \frac{3(1+\delta)}{1-\delta}  \Big( \frac{49}{36} -a \Big) + 4, \frac{b}{12} -\frac{10 \delta}{1-\delta}, \zeta\delta \Big\}>0.$$This completes the proof of Lemma \ref{thm:screened-it}.
\end{proof}

\section{Proof of the main Theorem} \label{sec:proof-main-result}

Now we prove the uniform bound $N\le Z+C$. The argument is identical to the proof of Theorem 1 in \cite{FraNamBos-16}, but we repeat it below for the convenience of the reader. 

\begin{proof}[Proof of Theorem \ref{main}] Since we have proved $N\le 2Z + C(Z^{2/3}+1)$ in Lemma \ref{roughbound}, it remains to consider the case $N\ge Z\ge 1$. By Theorem~\ref{thm:screened-intro}, we can find universal constants $C, \eps, D>0$ such that 
$$
|\Phi_{|x|}(x)  - \Phi^{\rm TF}_{|x|}(x)  | \le C |x|^{-4+\eps}, \quad \forall |x|\le D.
$$
In particular, \eqref{eq:assume-D} holds true with a universal constant $\beta=C D^\eps$. We can choose $D$ sufficiently small such that $D\leq 1$ and $\beta\leq 1$, which allows us to apply Lemma \ref{lem:screened-easy-bounds}. Then using \eqref{eq:int-rho-1} and \eqref{eq:int-rho-4} with $r=D$, we find that 
$$
\int_{|x|>D} \rho_0 + \left| \int_{|x|<D} (\rho_0 - \rho^{\rm TF}) \right|   \le C.
$$
Combining  with $\int\rho^{\rm TF}=Z$, we obtain the ionization bound
\begin{equation*}
N = \int \rho_0 = \int_{|x|>D} \rho_0 + \int_{|x|<D} (\rho_0-\rho^{\rm TF}) + \int_{|x|<D} \rho^{\rm TF} \le C + Z. \qedhere
\end{equation*}
\end{proof}

The proofs of Theorems \ref{thm:screened-intro} and \ref{thm:radius} follow from Lemma \ref{lem:screened} in the same way as \cite[Theorems 1 and 2]{FraNamBos-16} follow from \cite[Lemma 15]{FraNamBos-16}.



\end{document}